\documentclass{vldb}
\usepackage{graphicx}
\usepackage{balance}  

\vldbTitle{Differentially Private Hierarchical Count-of-Counts Histograms}
\vldbAuthors{}
\vldbDOI{https://doi.org/10.14778/3236187.3236202}
\usepackage{graphicx,subfigure}
\usepackage{color}
\usepackage[utf8]{inputenc}
\usepackage[hidelinks]{hyperref}
\usepackage[capitalise]{cleveref}
\usepackage[linesnumbered,ruled,vlined]{algorithm2e}

\usepackage{booktabs}  
\usepackage{amsmath}
\usepackage{comment}
\usepackage{paralist}
\usepackage{wrapfig}
\usepackage{enumitem}

\newtheorem{problem}{Problem}
\newtheorem{theorem}{Theorem}

\newtheorem{lemma}{Lemma}

\newtheorem{definition}{Definition}

\newcommand{\set}[1]{\ensuremath{\left\{#1\right\}}}
\DeclareMathOperator{\child}{children}

\newcommand{\coco}{count-of-counts\xspace}
\newcommand{\Coco}{Count-of-counts\xspace}
\newcommand{\changes}[1]{\textcolor{black}{#1}}

\allowdisplaybreaks

\begin{document}


\title{Differentially Private Hierarchical Count-of-Counts Histograms}



%
%
%
%

\numberofauthors{1} 

\author{
\alignauthor
\hspace{-6pt}Yu-Hsuan Kuo, Cho-Chun Chiu, Daniel Kifer$^\dagger$, 
Michael Hay$^\ddagger$, Ashwin Machanavajjhala$^\star$\\
\affaddr{Penn State University, $^\dagger$ Penn State University and U.S. Census Bureau,}\\
\affaddr{$^\ddagger$ Colgate University, $^\star$ Duke University}
\email{ \{yzk5145,dkifer\}@cse.psu.edu, cuc496@psu.edu, \\
mhay@colgate.edu, ashwin@cs.duke.edu}
%
%
}

\maketitle


\begin{abstract}
\changes{
We consider the problem of privately releasing a class of queries that we call \emph{hierarchical \coco histograms}. \Coco histograms partition the rows of an input table into groups (e.g., group of people in the same household), and for every integer $j$ report the number of groups of size $j$. Hierarchical \coco queries report \coco histograms at different granularities as per hierarchy defined on an attribute in the input data (e.g., geographical location of a household at the national, state and county levels).} In this paper, we introduce this problem, along with appropriate error metrics and propose a differentially private solution that generates \changes{\coco} histograms that are consistent across all levels of the hierarchy.
\end{abstract}

\section{Introduction}\label{sec:intro}

The publication of differentially private tables is an important area of study with applications
to statistical government agencies, like the U.S. Census Bureau, that collect and publish economic
and demographic data about the population. Most work has focused on ordinary histograms -- for example, generating counts of
how many people are in each combination of age, state, and business sector.

However, an important class of queries that has been under-studied are \changes{\emph{hierarchical \coco} histograms}, which are used to study the skewness of a distribution. The 2010 Decennial Census published 33 tables related to such queries \cite{sf1}, but these tables were truncated because formal privacy methods for protecting such tables did not exist. To get a \changes{\coco} histogram, \changes{one first aggregates records in a table \texttt{R} into groups (e.g., \texttt{A=SELECT groupid, COUNT(*) AS size FROM R GROUPBY groupid}) and then forms a histogram on the groups (\texttt{H=SELECT size, COUNT(*) FROM A GROUPBY size}}).
\changes{Thus $H$ can be treated as an array, where $H[i]$ is the number of groups of size $i$. When there is a hierarchical attribute associated with each group such as location, the goal is to estimate a histogram $H$ for every element in the hierarchy and enforce consistency: the sum of the histograms at the children equals the histogram at the parent.}
\changes{For example, consider a table \textbf{Persons}(person$\_$id, group$\_$id, location),  with hierarchy national/state/county on the location attribute. A hierarchical \coco histogram on this table would ask: for each geographic region (national, state, county) and every $j$, how many households (i.e. groups) in that region have $j$ people.}


\changes{Closely related to \coco histograms are}
unattributed histograms (also known as frequency lists) \cite{Hay2010Boosting,BDB16}, \changes{which are not hierarchical, and whose structure is not well-suited to express hierarchical constraints on a secondary attribute like location.} Unattributed histograms report the number of rows in the smallest group, followed by the number of rows in the second smallest group, etc., \changes{i.e. \texttt{$H_g$=SELECT COUNT(*) AS size FROM R GROUPBY groupid ORDERBY size}. One can convert unattributed histograms into \coco histograms and vice-versa, so differentially private unattributed histograms 
\cite{Hay2010Boosting,Lin2013Information} could be used to generate differentially private \coco histograms (and vice versa, see \Cref{subsec:unatt}).}  

However, the main focus of our work is \emph{hierarchies} \changes{in which every group (e.g., household) belongs to exactly one leaf (e.g., county).\footnote{\changes{Thus group$\_$id$\rightarrow$location is a functional dependency.}}
Such  hierarchies are natural and important for \coco histograms -- all of the queries in the 2010 Census were hierarchical \changes{(by location).} However, they are  not well supported by existing methods}. 
\changes{Currently, the only known way to get \emph{consistent} differentially private \coco or unattributed histograms at every level of a hierarchy is to estimate them only at the leaves and then aggregate them up the hierarchy. However, our experiments show this approach introduces high error at non-leaf nodes (like in other hierarchical problems \cite{Hay2010Boosting,Qardaji2013Hierarchical}).}

\newlength{\tmpcolumnsep}
\setlength{\tmpcolumnsep}{\columnsep}
\setlength{\columnsep}{10pt}
\newlength{\tmpintextsep}
\setlength{\tmpintextsep}{\intextsep}
\setlength{\intextsep}{4pt}

\changes{An alternative is to separately obtain  differentially private estimates at every node in the hierarchy and then postprocess them to be consistent (i.e. aggregating information from the children should produce the corresponding histogram at} 
\begin{wraptable}{r}[0pt]{0cm}
\begin{tabular}{|c|c|c|}\hline
\multicolumn{1}{|c}{\textbf{g$\_$id}} & \multicolumn{1}{c}{\textbf{size}} &\multicolumn{1}{c|}{\textbf{loc.}}\\\hline
1 & 4 & a\\
2 & 2 & b\\
3 & 1 & a\\
4 & 1 & b\\\hline
\end{tabular}
\end{wraptable}
\changes{ the parent). However, this is far from trivial. To see why, suppose we aggregate the Persons table by group id (g$\_$id) to obtain the following table. At the root, the \coco histogram is $H^{\text{top}}=[2, 1, 0, 1]$ (i.e., 2 groups of size one, 1 group of size two, 0 of size three, 1 group of size four) and the  corresponding unattributed histogram is $H_g^{\text{top}}=[1,1,2,4]$. At node $a$, the \coco histogram is $H^a=[1, 0, 0, 1]$ and unattributed histogram is $H_g^a=[1,4]$; at node $b$ they are $H^b=[1,1,0,0]$ and $H_g^b=[1,2]$, respectively.} 

\setlength{\intextsep}{\tmpintextsep}
\setlength{\columnsep}{\tmpcolumnsep}

\changes{The unattributed histogram is not additive: $H_g^{\text{top}}\neq H_g^a + H_g^b$, thus existing techniques for enforcing consistency in a hierarchy\cite{Hay2010Boosting,Qardaji2013Hierarchical} do not apply. The \coco histograms are additive, but consistency is still nontrivial for the following reasons.}
%
The standard approach is to formulate ``consistency'' as an optimization problem \cite{Hay2010Boosting}.\footnote{\changes{In this case, given differentially private estimates of \coco histograms, modify them as little as possible so that children add up to their parents.}} Generic solvers are slow due to the number of variables involved so fast specialized algorithms, like mean-consistency  \cite{Hay2010Boosting,Qardaji2013Hierarchical}, have been proposed. However, mean-consistency cannot solve our problem. First, it can produce negative and fractional answers, both of which are invalid query answers \changes{(for a full list of requirements, see Section \ref{sec:notation})}.\footnote{Negative answers can be obtained from the final step of the mean-consistency algorithm that subtracts a constant from the counts at each node's children.} Second, it relies on an estimate of query variances that are difficult to obtain. Instead, we propose a different consistency algorithm based on an efficient and optimal matching of groups at different levels of the hierarchy. 

Aside from consistency, there are other challenges for differentially private hierarchical \changes{\coco} queries. One such challenge is properly designing the error metric (this is true even without a hierarchy). Standard measure like $L_1$ or $L_2$ (sum-squared error) distance between the true \changes{\coco} histogram $H$ and the differentially private histogram $\widehat{H}$ are inapplicable. For instance, suppose the true data $H$ had 20 groups and all had size 1. Consider two different estimates. Estimate $\widehat{H}_A$ has 20 groups, all having size 2. Estimate $\widehat{H}_B$ has 20 groups all having size 10. $\widehat{H}_A$ and $\widehat{H}_B$ both have the same $L_1$ error and $L_2$ error, but clearly $\widehat{H}_A$ is better than $\widehat{H}_B$ because its groups are closer in size to $H$. Earth-mover's distance \cite{emd} is more appropriate and turns out to be efficient to compute for our problem.

Another challenge is obtain\changes{ing} the initial differentially private \coco estimates at each node of the hierarchy (i.e. the initial estimates on which we will run consistency algorithms). Naively adding noise to each \changes{element of $H$} results in poor performance. However, we show that some unattributed histogram algorithms \cite{Hay2010Boosting,Lin2013Information} can be adapted for this task. \changes{Nevertheless, most of the time a different approach based on \coco histograms works better empirically.}\footnote{
\changes{Of course, one could also use generic tools like Pythia \cite{Kotsogiannis2017:Pythia} or the technique of Chaudhuri et al. \cite{Chaudhuri2011ERMJMLR} for selecting between the two approaches for estimating count-of-count histograms at each level in the hierarchy.}}
%
To summarize, our contributions are:
\begin{itemize}[noitemsep,leftmargin=5mm]
\item We introduce the \emph{hierarchical} differentially private \changes{\coco} histogram problem.
\item We propose \changes{new and accurate} algorithms for the non-hierarchical version of this problem.
\item For the hierarchical version, we propose algorithms that force consistency between estimates at different levels of the hierarchy.
\item An evaluation on a combination of real and partially-synthetic data validates our approach. The partially synthetic data are used to extend a real dataset, in which group sizes were truncated at a small value because, at the time, it was not known how to publish full results while protecting privacy.
\end{itemize}

This paper is organized as follows. We review related work in Section \ref{sec:related}. We formally define the problem and notation in Section \ref{sec:notation}. We propose algorithms for the non-hierarchical version of the problem in Section \ref{sec:nonhier}. We show how to obtain consistency in the hierarchical version in Section \ref{sec:consistency}. We present experiments in Section \ref{sec:experiments}. We discuss conclusions and future work in Section \ref{sec:future}.

\section{Related Work}\label{sec:related}

Differentially private histograms have been the main focus of private query answering algorithms.
Starting with the simplest application, adding Laplace noise to each count \cite{DMNS06Calibrating}, these methods have progressively become more sophisticated and data aware to the point where they take advantage of structure in the data (such as clusters) to improve accuracy (eg., \cite{Acs2012Histogram,Li2014DAWA,Zhang2014AHP,Xu2012Histogram}). Extensions for optimizing various count queries have also been proposed \cite{Li2014DAWA,Li2010MM,Hardt2012MWEM}.


Ordinary histograms have also been extended with hierarchies. The addition of hierarchies makes data release problems more challenging as well as more applicable to real-world uses. One of the earliest examples of hierarchical histograms was introduced by Hay et al. \cite{Hay2010Boosting}, in which the data consist of a one-dimensional histogram that is converted to a tree, where each node represents a range  and stores one value (count of the data points in that range) and the value at a node must equal the sum of the values stored by the children. Qardaji et al. \cite{Qardaji2013Hierarchical} determined a method to compute the fanout of the tree that is approximately optimal for answering range queries. Other, more flexible, partitioning of the data have also been studied (e.g., \cite{Zhang2016PrivTree,Cormode2011Spatial,Xiao2010Partitioning}). Ding et al. \cite{Ding2011DPCube} provided extensions to lattices in order to answer data cube queries instead of just range queries. In our setting, the hierarchies are already provided and the goal is to get consistent \changes{\coco} histograms at each level of the hierarchy. The algorithms used for consistency in ordinary histograms do not satisfy the requirements of \changes{\coco} histograms (as explained in Sections \ref{sec:notation} and \ref{sec:consistency}), and so new consistency algorithms are needed.

The work most closely related to our problem are unattributed histograms \cite{Hay2010Boosting,BDB16} which are often used to study degree sequences in social networks \cite{Lin2013Information,Karwa2016Inference}. Unattributed histograms are the duals of \changes{\coco} queries (they count people rather than groups) and can be used to answer queries such as ``what is the size of the $k^\text{th}$ largest group?'' They are much more accurate than the naive strategy of adding noise to each group and selecting the $k^\text{th}$ largest noisy group \cite{Hay2010Boosting}. Although unattributed histograms do not have a hierarchical component, \changes{our techniques solve the hierarchical version of this problem because \coco histograms can be converted to unattributed histograms}. 

\section{Problem Definition}\label{sec:notation}
Consider a database $D$ consisting of these 3 tables: Entities(entity\_id, group\_id), Groups(group\_id, region\_id), and Hierarchy(region\_id, level$_0$, level$_1$,\dots, level$_{L}$), where the entity\_id and group\_id are randomly generated unique numbers.
 Every entity belongs to a group and every group is in a region. 
Regions are organized into a hierarchy (as encoded by the table Hierarchy), where level $0$ is the root containing all regions, level 1 subdivides level 0 into a disjoint set of subregions, and, recursively, level $i+1$ subdivides the regions in level $i$. For example, level $0$ can be an entire country, level $1$ can be the set of states, and level $2$ can be the set of counties. We let $\Gamma$ represent the hierarchy and we will use $\tau$ to denote a node in the hierarchy. For each region $\ell$, we let Level$_j(\ell)$ denote its ancestor in level $j$. For example, if $\ell$ is a region corresponding to Fairfax County, then Level$_1(\ell)=$``Virginia'' and Level$_2(\ell)=$``Fairfax County.''
We add the restriction that a group cannot span multiple leaves of the region hierarchy (i.e. each group is completely within the boundaries of a leaf node).

We now consider what information is public and what information is private.
\begin{itemize}[noitemsep,leftmargin=5mm]
\item Hierarchy -- this table only defines region boundaries and so is considered \underline{public}.
\item Groups -- the group id is a random number so the only information this table provides is how many groups are in each region. We consider this table to be \underline{public} to be consistent with real world applications such as at the U.S. Census Bureau, where the number of households and group quarters facilities in each Census block is assumed (by the Bureau) to be public knowledge because it is easy to obtain by inspection.\footnote{If one wishes to make the Groups table private, our methods can be extended. The most straightforward approach is to first estimate the number of groups in each region by adding Laplace noise to each count. These estimates can be made consistent by solving a nonnegative least squares optimization problem. Since there is only one number per region, it is a relatively small problem that can be solved with off-the-shelf optimizers. Once the counts are generated they can be used with our algorithm.}
\item Entities -- this table contains information about which entities (e.g., people) are in the same group. We treat it as \underline{private}.
\end{itemize}
For example, in Census data, groups can be housing facilities (households and group quarters) and entities are people. In taxi data, groups could correspond to taxis and entities to pick-ups of passengers.

Each node $\tau$ in the hierarchy $\Gamma$ has an associated group-size histogram $\tau.H$ (or simply $H$ when $\tau$ is understood from the context). $\tau.H[i]$ is the number of groups, in the region associated with $\tau$, that have $i$ entities in them. We also use $\tau.G$ to represent the (public) number of groups in $\tau$ since that can be derived from the public Groups table.
There are two other convenient representations of the \changes{\coco} histogram $\tau.H$:
\begin{itemize}[noitemsep,leftmargin=5mm]
\item $\mathbf{\tau.H_c}$: This is the \emph{cumulative sum histogram}, defined as $\tau.H_c[i]=\sum_{j=0}^i \tau.H[j]$, which is the number of groups of size less than or equal to $i$. Note that the last element of $\tau.H_c$ is therefore $\tau.G$ (the total number of groups in $\tau$). For example, if $\tau.H=[0,2,1,2]$ then $\tau.H_c=[0, 2, 3, 5]$.
\item $\mathbf{\tau.H_g}$: This is the \emph{\changes{unattributed} histogram}. $\tau.H_g[i]$ is the size of the $i^\text{th}$ smallest group in $\tau$. Thus the dimensionality of $\tau.H_g$ is $\tau.G$. Note that $\tau.H_g$ is an unattributed histogram in the terminology of Hay et al. \cite{Hay2010Boosting}. For example, if $\tau.H=[0,2,1,2]$ then $\tau.H_g=[1,1,2,3,3]$ (since there are two groups of size 1, one of size 2, and two of size 3).
\end{itemize}

The conversion in representation from $H$ to $H_g$ to $H_c$ (and vice versa) is straightforward and is omitted.

The privacy definition we will be using is $\epsilon$-differential privacy \cite{DMNS06Calibrating} applied at the entity level. More specifically:
\begin{definition}[Differential Privacy]
Given a privacy loss budget $\epsilon>0$, a mechanism $M$ satisfies $\epsilon$-differential privacy if, for any pair of  databases $D_1$, $D_2$ that contain the public Hierarchy and Groups tables, and differ by the presence or absence of one record in the Entities table, and for any possible set $S$ of outputs of $M$, the following is true:
$$P(M(D_1)\in S) \leq e^\epsilon P(M(D_2)\in S)$$
\end{definition}

Thus the differentially private hierarchical \changes{\coco} problem can be defined as:
\begin{problem}
Let $\Gamma$ be a hierarchy such that each node $\tau\in \Gamma$ has an associated \changes{\coco} histogram $\tau.H$. Develop an algorithm to release a set of estimates $\set{\tau.\widehat{H}~:~\tau\in \Gamma}$ while satisfying $\epsilon$-differential privacy along with the following desiderata:
\begin{itemize}[itemsep=0em,leftmargin=5mm]
\item \textbf{[Integrality]}: $\tau.\widehat{H}[i]$ is an integer for all $i$ and $\tau\in\Gamma$.
\item \textbf{[Nonnegativity]}: $\tau.\widehat{H}[i]\geq 0$ for all $i$ and $\tau\in\Gamma$.
\item \textbf{[Group Size]}: $\sum_i \tau.\widehat{H}[i] = \tau.G$ for all $\tau\in \Gamma$.
\item \textbf{[Consistency]}: $\tau.\widehat{H}[i] = \hspace{-0.3cm}\sum\limits_{c\in \child(\tau)} \hspace{-0.45cm} c.\widehat{H}[i]$ for any $\tau\in\Gamma$, $i$. 
\end{itemize}
\end{problem}
These constraints ensure that the \changes{\coco} histograms satisfy all publicly known properties of the original data.

\subsection{Error Measure}\label{subsec:error}

The error measure is an important aspect of the problem as we would like to quantify the ``distance'' between $\tau.H$ and $\tau.\widehat{H}$. 
Ideally, we would like to measure this distance as the minimum number of people that must be added or removed from groups in $\tau.H$ to get $\tau.\widehat{H}$.

The standard error measures of Manhattan distance $||\tau.H - \tau.\widehat{H}||_1$ or sum-squared error $||\tau.H - \tau.\widehat{H}||_2^2$ do not capture this distance measure. To see why, suppose $H=[0, 100, 0,$ $ 0, 0, 0]$ meaning that all 100 groups have size $1$. Consider two estimates $\widehat{H}_1=[0, 0, 100, 0, 0, 0, 0]$ (where all groups have size 2) and $\widehat{H}_2=[0, 0, 0, 0, 0, 100]$ (where all groups have size 5). We see that $||H - \widehat{H}_1||_1=||H-\widehat{H}_2||_1=200$ and $||H - \widehat{H}_1||_2^2=||H-\widehat{H}_2||_2^2=20,000$. However,  we should consider $\widehat{H}_1$ to be closer to $H$ than $\widehat{H}_2$ is. The reason is that if we add one extra person into each group in $H_1$, we would obtain $\widehat{H}_1$. On the other hand, to obtain $\widehat{H}_2$, we would need to add 4 people to each group.

Thus, the appropriate way to measure distance between $H$ and $\widehat{H}$ is the Earthmover's distance (emd) as it precisely captures the number of people that must be added to or removed from groups in $H$ in order to obtain $\widehat{H}$. Normally, computing emd is not linear in the size of an array \cite{emd}. However, it can be computed in linear time for our problem by using the cumulative histograms.

\begin{lemma}[\cite{tcloseness}]\label{lem:error}
The earthmover's distance between $H$ and $\widehat{H}$ can be computed as $||H_c - \widehat{H}_c||_1$, where $H_c$ (resp., $\widehat{H}_c$) is the cumulative histogram of $H$ (resp., $\widehat{H}$). It is the same as the $L_1$ norm in the $H_g$ representation when the number of groups is fixed.
\end{lemma}

Our algorithms optimize error according to this metric.

\subsection{Privacy Primitives}
We now describe the privacy primitives which serve as building blocks of our algorithm. One important concept is the sensitivity of a query.
\begin{definition}
Given a query $q$ (which outputs a vector), the global sensitivity of $q$, denoted by $\Delta(q)$ is defined as:
$$\Delta(q)=\max_{D_1,D_2}||q(D_1)-q(D_2)||_1,$$
where the maximum is taken over all databases $D_1, D_2$ that contain the public Hierarchy and Groups tables, and differ by the presence or absence of one record in the Entities table. 
\end{definition}

The sensitivity is used to calibrate the scale of noise needed to achieve differential privacy. We can use the Geometric Mechanism \cite{Ghosh2009Universal} instead of the Laplace Mechanism \cite{DMNS06Calibrating}, because we want our final counts to be integers. The Geometric mechanism is also preferable to the Laplace mechanism as it has lower variance, and is not susceptible to side-channel attacks when implemented in floating point arithmetic \cite{mironov12}. 

\begin{definition}[Geometric Mechanism \cite{Ghosh2009Universal}]\label{def:geomech}

Given a \\ 
database $D$, a query $q$ that outputs a vector, and a privacy loss budget $\epsilon$, the geometric mechanism adds independent noise to each component of $q(D)$ using the following distribution: $P(X=k) = \frac{1-e^{-\epsilon}}{1+e^{-\epsilon}} e^{-\epsilon |k|/\Delta(q)}$ (for $k=0, \pm 1, \pm 2,$ etc.). This distribution is known as the double-geometric with scale $\Delta(q)/\epsilon$.
\end{definition}

\begin{lemma}[\cite{DMNS06Calibrating,Ghosh2009Universal}] The Geometric Mechanism satisfies $\epsilon$-differential privacy.
\end{lemma}

\section{The Base Case: Non-hierarchical Count-of-Counts}\label{sec:nonhier}
In order to create consistent estimates of $\tau.H$ for all nodes $\tau$ in the hierarchy $\Gamma$, we first create estimates of $\tau.H$ independently and then post-process them for consistency. In this section, we discuss how to generate a differentially private estimate $\tau.\widehat{H}$ for a single node (i.e. temporarily ignoring the hierarchy). In Section \ref{sec:consistency}, we  show how to combine the estimates $\tau.\widehat{H}$ (for all $\tau\in \Gamma$) to satisfy consistency. 

In the rest of this section, we focus on a single node in $\Gamma$, so to simplify notation we use the notation $H$ instead of $\tau.H$. We first discuss a naive strategy for estimating $\widehat{H}$, along with 2 more robust strategies. 

\subsection{Naive Strategy}\label{subsec:naive}
The naive strategy for estimating $\widehat{H}$ is to use Definition \ref{def:geomech} and add double-geometric noise with scale $2/\epsilon$ to each cell of $H$. However, because the maximum group size is not known, the length of $H$ is private. Thus we must determine a maximum non-private size $K$ and make the following modifications to $H$. If all groups in $H$ have less than $K$ people, then $H$ is extended with 0's until its length is $K+1$. If some groups of $H$ have size more than $K$, we change the sizes of those groups to $K$. Call the resulting histogram $H^\prime$. We then add double-geometric noise with scale $2/\epsilon$ to each cell of $H^\prime$.
This strategy satisfies $\epsilon$-differential privacy because the sensitivity of $H^\prime$ is 2:
\begin{lemma}
The global sensitivity of $H^\prime$ is 2.
\end{lemma}
\begin{proof}
For any group size $i<K$, adding a person to a group of size $i$ decreases $H^\prime[i]$ by one and increases $H^\prime[i+1]$ by one. When $i\geq K$, there is no change to $H^\prime$. Similarly, removing a person from a non-empty group of size $i\leq K$ increases $H^\prime[i-1]$ by one and decreases $H^\prime[i]$ by one (for a total change of 2). If $i>K$,  there is no change to $H^\prime$. 
\end{proof}

Let $\widetilde{H}$ be the noisy version of $H^\prime$. The numbers in $\widetilde{H}$ can be negative. Thus we post-process $\widetilde{H}$ to obtain an estimate $\widehat{H}$ by solving the following problem (using a quadratic program solver):
\begin{align*}
\widehat{H} &=\arg\min_{\widehat{H}} ||\widetilde{H} - \widehat{H}||_2^2\\
&\text{ s.t. }\widehat{H}[i]\geq 0 \text{ for all i}\quad \text{ and }\sum_i \widehat{H}[i] = G
\end{align*}
To get integers, we set $r = G-\sum_i \lfloor \widehat{H}[i] \rfloor$, round the cells with the $r$ largest fractional parts up, and round the rest down.


\changes{This approach has several weakness. First, there are many indexes $i$ where $H[i]=0$ and, after noise addition, many of them (a $\frac{2e^{-\epsilon}}{e^{-\epsilon}+1}$ fraction of indices) will have non-zero entries. Then roughly half of them end up with positive counts due to nonnegativity constraints. As a result, in our experiments, this method had several orders of magnitude worse error than the algorithms we describe next. Second, earthmover's distance is equivalent to measuring error between the cumulative sums of the estimate $\widehat{H}$ and the original data $H$: $\sum_{i} | \sum_{j\leq i} \widehat{H}[j] - \sum_{j\leq i} H[j]|$. Since noise is added to every cell independently, the $i^\text{th}$ component of the error depends on the sum of $i$ random variables (whose total variance is $O(i)$). Hence, if $H$ has $n$ components, one would expect the error to be $O(n^2)$.}

\subsection{Unattributed Histograms $H_g$}\label{subsec:unatt}

The next approach is to use algorithms for unattributed histograms \cite{Hay2010Boosting}. We can convert $H$ into the representation $H_g$, where $H_g[i]$ is the size of the $i^\text{th}$ smallest group. The length of $H_g$ is $G$, so potentially this can be a very large histogram. One of the properties of $H_g$ is that it is non-decreasing. Hence we can achieve $\epsilon$-differential privacy by adding independent double-geometric noise with scale $1/\epsilon$ to each element of $H_g$ to obtain $\widetilde{H}_g$, because the sensitivity of unattributed histograms is 1 \cite{Hay2010Boosting}. However, $\widetilde{H}_g$ is no longer non-decreasing (or even nonnegative), hence, following \cite{Hay2010Boosting,Lin2013Information}, we post-process it by solving the following optimization problem with either $p=1$ or $p=2$:
\begin{align*}
\widehat{H}_g &=\arg\min_{\widehat{H}_g} ||\widetilde{H}_g-\widehat{H}_g||_p^p\\
&\text{ s.t. } 0\leq \widehat{H}_g[i]\leq \widehat{H}_g[i+1] \text{ for }i=0,\dots, G-1
\end{align*}
Then we round each entry of $\widehat{H}_g$ to the nearest integer. Since the result was non-decreasing before rounding, it will remain non-decreasing after rounding.  Then from the resulting $\widehat{H}_g$ estimate, we convert it back to $\widehat{H}$ by counting, for each $i$, how many estimated groups have size $i$.

This optimization problem is known as \emph{isotonic regression}. When $p=2$ it can be solved in linear time using the min-max algorithm \cite{barlow1972isotonic}, pool-adjacent violators (PAV) \cite{barlowstatistical,robertson1968estimating}, or a commercial optimizer such as Gurobi \cite{gurobi}. When $p=1$ it can be converted to a linear program but it runs much slower. In our experiments,
we used $p=2$ because $H_g$ can have length in the hundreds of millions. For these sizes, the quadratic program is much faster to solve using PAV.

One observation we had is that this method is very good at estimating large group sizes, and most of its error comes from estimation errors of small group sizes (see Figure \ref{fig:abs_err}).

\subsection{Cumulative Histograms $H_c$}\label{subsec:cumu}
Another strategy is to use $H_c$, the cumulative sum of $H$. Since the cumulative sum is non-decreasing, we again can add noise and use isotonic regression. As in the naive case, we must determine a public upper bound\footnote{\changes{Recall $K$ is an upper bound on the maximum number of people in a group. This method is not very sensitive to K -- in the experiments we used $K=100,000$ on datasets where the largest group had around 10,000 people -- an order of magnitude difference and still the estimated size of the largest group ended up being around 10,000. Thus if we have no prior knowledge, we can estimate $K$ as follows. Set aside a small privacy budget, since $K$ does not need much accuracy (e.g., $\epsilon=10^{-4}$). Let $X$ be the number of people in the largest group. Estimate $K$ as $X+$Laplace$(1/\epsilon) + 5\frac{\sqrt{2}}{\epsilon}$ -- i.e. add 5 standard deviations to a noisy estimate of $X$ so that $P(K \geq X)> 0.9995$ }} $K$ on $H$ and modify $H$ appropriately (as in Section \ref{subsec:naive}) before computing the cumulative sum. We saw in Section \ref{subsec:error} that error in estimation of \changes{\coco histograms} is measured as the $L_1$ difference between cumulative size histograms. Thus it makes sense to privatize these histograms directly. 
\begin{lemma}\label{lem:hcsens}
The global sensitivity of $H_c$ is 1.
\end{lemma}
\begin{proof}
Adding one person to a group of size $i$ means there is one less group of size $i$ and one more group of size $i+1$. Thus $H_c[i]$ decreases by one but $H_c[i+1]$ remains the same (i.e., number of groups $\leq i+1$ does not change). None of the other entries of $H_c$ change. Similarly, removing a person from a group of size $i$ means $H_c[i-1]$ increases by one and $H_c[i]$ does not change (and neither do any other entries). Thus the overall change is $1$.
\end{proof}
Thus we can satisfy differential privacy by adding independent double-geometric noise with scale $1/\epsilon$ to each cell of $H_c$ to get $\widetilde{H}_c$. We then postprocess $\widetilde{H}_c$ by solving the following optimization problem (using $p=1$ or $p=2$):
\begin{align*}
\widehat{H}_c &=\arg\min_{\widehat{H}_c} ||\widehat{H}_c - \widetilde{H}_c||_p^p\\
&\text{ s.t. } 0 \leq \widehat{H}_c[i] \leq \widehat{H}_c[i+1] \text{ for }i=0,\dots, K\\
& \text{ and }\widehat{H}[K]=G
\end{align*}
These problems can again be solved using PAV (in the case of $p=2$) or with a commercial optimizer (in the cases of $p=1$ or $p=2$).
In our experiments, we found that the $L_1$ version of the problem (with $p=1$) performs better than the $L_2$ version (with $p=2$). This is consistent with prior observations on unattributed histograms \cite{Lin2013Information}. A Bayesian post-processing is known to further reduce error, but we did not use it because it scales quadratically with the size of the histogram (the sizes of our histograms make this prohibitively expensive) \cite{Lin2013Information}. 
Finally, we then round $\widehat{H_c}$ to nearest integer and convert it back into a \coco histogram $\widehat{H}$. We found that the $L_1$ version of the problem mostly returns integers, so rounding is minimal.

Unlike the $H_g$ method, we found that this method is accurate for small group sizes but less accurate for large groups.

\begin{figure}[t]
	\centering
	\includegraphics[width=7.5cm]{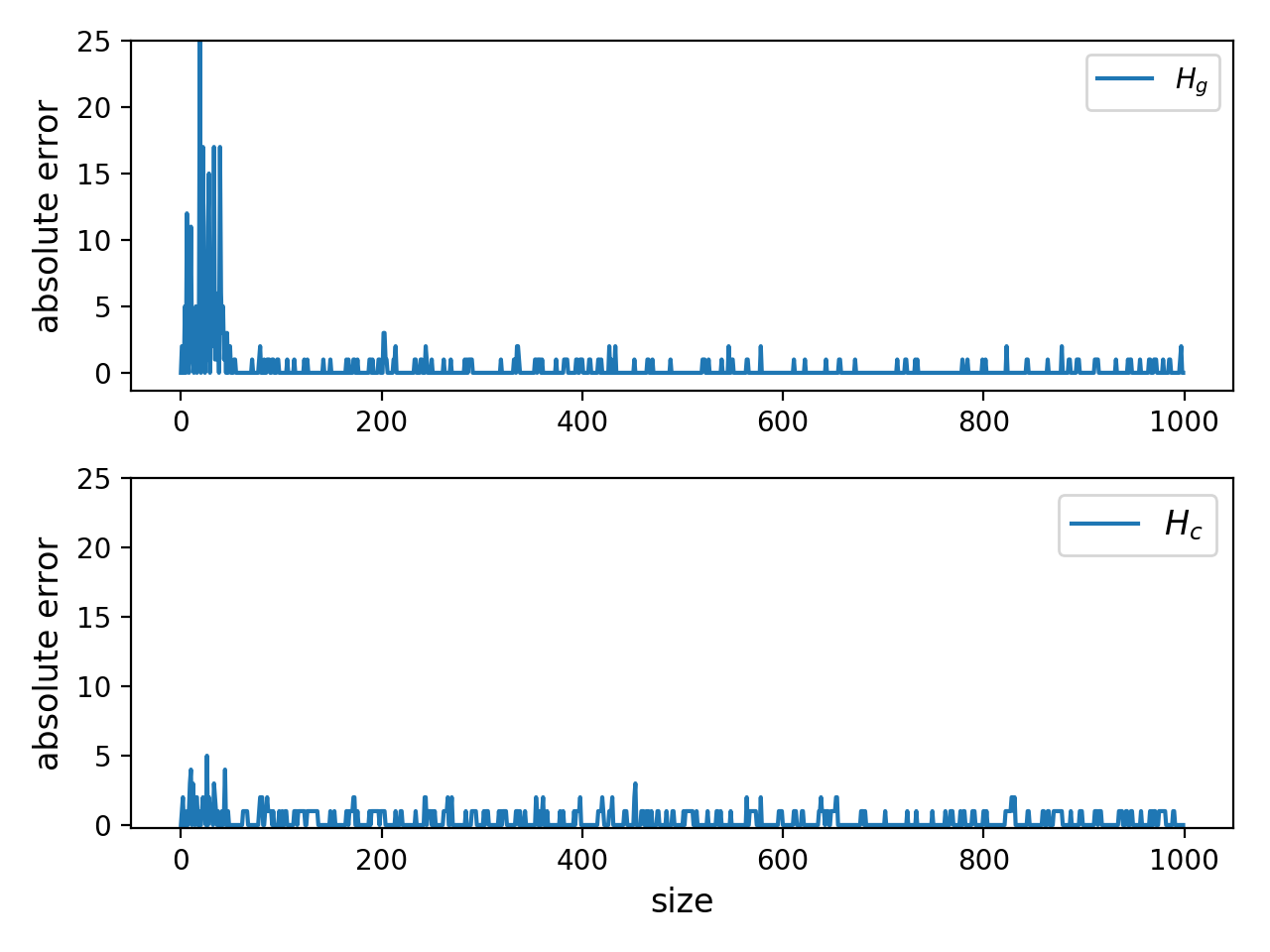}
\vspace{-9pt}
	\caption{Error visualization. $x$-axis: cumulative sum of group size counts. $y$-axis: estimation error at these sizes. Top: $H_g$ method (Section \ref{subsec:unatt}). Bottom $H_c$ method (\ref{subsec:cumu}). Errors from the $H_g$ method are concentrated around small group sizes. Errors from $H_c$ are more dense throughout the rest of the sizes.}
    \label{fig:abs_err}
\end{figure}

\section{Hierarchical Consistency}\label{sec:consistency}
Our overall algorithm for the hierarchical problem is shown in Algorithm \ref{alg:top-down}.
The hierarchy $\Gamma$ has $L+1$ levels (including the root), so we use $\epsilon/(L+1)$ budget for estimating  $\tau.\widehat{H}$ for each $\tau\in \Gamma$ (using either the $H_c$ method or $H_g$ method from Section \ref{sec:nonhier}). 
Then we run a consistency post-processing algorithm to obtain new histograms $\tau.\widehat{H}^\prime$ such that the histograms at each parent equals the sum of the histograms at its child nodes.

Typically, for a hierarchy $\Gamma$ with noisy data $\tau.\widehat{H}$ at each node, one would use the mean-consistency algorithm \cite{Hay2010Boosting} to get consistent estimates $\tau.\widehat{H}^\prime$ with the property that the sum of data at child nodes equals the data at the parent node. In the group-size estimation problem, this algorithm will not produce outputs satisfying the problem requirements (Section \ref{sec:notation}). The reason is the following. The mean-consistency algorithm solves the global optimization:
\begin{align*}
\set{\tau.\widehat{H}^\prime~:~\tau\in\Gamma} &=\arg\min_{\set{\tau.\widehat{H}^\prime~:~\tau\in\Gamma}} \sum\limits_{\tau\in\Gamma} ||\tau.\widehat{H}^\prime-\tau.\widehat{H}||_2^2\\
&\text{ subject to }\quad \tau.\widehat{H}^\prime =\sum_{c\in\child(\tau)} c.\widehat{H}^\prime
\end{align*}

This algorithm does not meet the necessary requirements for several reasons. First, it outputs real (and even negative\footnote{The solution given by mean-consistency can be negative even if all input numbers are positive. We verified this phenomenon using that algorithm and also with commercial optimizers. Intuitively, this happens because the mean-consistency algorithm \cite{Hay2010Boosting} has a subtraction step, in which a constant number is subtracted from each child so that they add up to the parent total. For children with small counts, this subtraction gives negative numbers.}) numbers while our problem requires the conditions that $\tau.\widehat{H}^\prime$ is an integer, is nonnegative, and $\sum_i \tau.\widehat{H}^\prime[i] = \tau.G$ (the private estimator must match the publicly known groups table).  Second, it requires knowledge of variances \cite{Qardaji2013Hierarchical} -- in our case, the variances of $\tau.\widehat{H}[i]$ for every $\tau$ and $i$. Not only is this variance  different for every $i$ and $\tau$, but it has no closed form (because of the isotonic regression used to generate $\tau.\widehat{H}$). 

Our proposed solution converts $\tau.\widehat{H}$ back into the \changes{unattributed} histogram $\tau.\widehat{H}_g$. It estimates the variance of the size of each group (Section \ref{subsec:initvar}); e.g., variance estimates for $\tau.\widehat{H}_g$, not $\tau.\widehat{H}$. It then finds a $1$-to-$1$ optimal matching between groups at the child nodes and groups at the parent node (Section \ref{subsec:match}). This means that each group has a size estimate from the child and a size estimate from the parent. It merges those two estimates (Section \ref{subsec:merge}). The result is a consistent set of estimates $\tau.\widehat{H}^\prime$ that satisfy all of the constraints in our problem. 
The overall approach that puts these pieces together is shown in Algorithm \ref{alg:top-down}, while the specifics are discussed next.

\begin{algorithm}[t]
\caption{Top-down Consistency}
\label{alg:top-down}
\SetNoFillComment
\DontPrintSemicolon
\KwIn{privacy loss budget $\epsilon$, Hierarchy $\Gamma$ with root at level 0}
\For{$\ell=0,\dots, L$\label{line:tdlevel}} {
   $\epsilon_\ell = \epsilon/(L+1)$\;
   \For{each node $\tau$ in level $\ell$ of $\Gamma$}{
      $\tau.\widehat{H}\gets $\text{$H_c$ or $H_g$ method}$(\epsilon_\ell,\tau.H, \tau.G)$\label{line:tdswitch}\tcp*{Sec \ref{sec:nonhier}}
      convert $\tau.\widehat{H}$ to $\tau.\widehat{H}_g$\;
   }

}
\tcc{Consistency Step}
\For{every node $\tau$ in $\Gamma$}{
  $\tau.V_g \gets $EstVariance($\tau.\widehat{H}_g$)\label{line:estvar}\tcp*{Section \ref{subsec:initvar}}
}
root$.\widehat{H}^\prime = root.\widehat{H}$; root$.V^\prime_g=$root$.V_g$\;
\For{$\ell=0,\dots,L-1$}{
	\For{every node $\tau$ in level $\ell$ of $\Gamma$}{
            \tcc{Matching alg in Section \ref{subsec:match}}
			 $m_\tau$ $\gets$ Match($\tau.\widehat{H}$,$\set{c.\widehat{H}~:~c\in\child(\tau)}$)
             \For(\tcc*[h]{Sec \ref{subsec:merge}}){each $c\in\child(\tau)$}{
                $(c.\widehat{H}^\prime_g, c.V^\prime_g) \gets $Update$(m_\tau, c.\widehat{H}_g, c.V_g,\tau.\widehat{H}^\prime_g,\tau.V^\prime_g)$
             }             
%
%
%
%

	}
}
Convert $\tau.\widehat{H}^\prime_g$ to $\tau.\widehat{H}^\prime$ for all leaf nodes $\tau$\;
\For(\tcc*[h]{Back-Substitution}){$\ell = L-1,\dots, 0$}{
  $\tau.\widehat{H}^\prime = \sum_{c\in\child(\tau)} c.\widehat{H}^\prime$
}
\Return $\set{\tau.\widehat{H}^\prime~:~\tau\in\Gamma}$
\end{algorithm}

\subsection{Initial Variance Estimation}\label{subsec:initvar}
The first part of our algorithm produces a differentially private \changes{\coco} histogram $\tau.\widehat{H}$ for every node $\tau$. We then convert it into the \changes{unattributed} histogram $\tau.\widehat{H}_g$. For each  $i$, we need an estimate of the variance of the $i^\text{th}$ largest group $\tau.\widehat{H}_g[i]$. Now, the \changes{\coco} histogram $\tau.\widehat{H}$ is obtained either from the  $H_g$ method (Section \ref{subsec:unatt}) or the cumulative histogram $H_c$ method (Section \ref{subsec:cumu}) and so the variance depends on which method was used.\footnote{\changes{Generally, $H_c$ works well for all levels. Users preferring fine-grained control can use generic algorithm selection tools \cite{Kotsogiannis2017:Pythia,Chaudhuri2011ERMJMLR}.}}

\subsubsection{Variance estimation for the $H_g$ method} \label{subsubsec:HgVar}
Recall that in the \changes{unattributed} histogram method, we obtained a noisy array (noise was added to $\tau.H_g$) and performed isotonic regression on it to get $\tau.\widehat{H}_g$. Since we need an estimate of the variance  of $\tau.\widehat{H}_g[i]$, we can use the following special properties of isotonic regression solutions  \cite{barlowstatistical}. As shown in Figure \ref{fig:l2}, isotonic regression is equivalent to taking a noisy 1-d array, partitioning the array, and assigning the same value to each element within a partition. In the case of $L_2$ isotonic regression, this value is the average of the noisy counts from the partition. In the case of $L_1$, this value is the median of the noisy counts in the partition. Thus the variance of each cell in the resulting array depends on the variance due to partitioning and the variance due to averaging the noisy counts. We cannot quantify the variance due to partitioning. However, the variance due to averaging noisy counts can be estimated. 

Each noisy count is generated by adding noise from the double-geometric distribution with scale $1/\epsilon_1$. Its variance can be approximated by the variance of the Laplace distribution with the same scale. Namely, the variance is $2/\epsilon_1^2$. In a partition of size $S$ (and $L_2$ isotonic regression), the value assigned to that partition is the average\footnote{In the case of $L_1$ regression, the value assigned to the partition is the median of $S$ noisy values; the variance of the median is difficult to compute, so we again estimate it as $2/(S\epsilon^2)$} of $S$ noisy values and so has variance $2/(S\epsilon^2)$. The partitions that were created by the solution of the isotonic regression are easy to determine -- they are simply the consecutive entries in the solution that have the same value.

Thus, our estimate of the variance for group $i$ in $\tau.\widehat{H}_g$ is computed as follows. Let $S_i$ be the number of groups that were in the same partition as $i$ in the solution (i.e. the number of entries in $\tau.\widehat{H}_g$ that equal $\tau.\widehat{H}_g[i]$). Set the variance estimate for the $i^\text{th}$ largest group to be  the following in Line \ref{line:estvar} of Algorithm \ref{alg:top-down}: $\tau.V_g[i]= \frac{2}{|S_i|\epsilon_1^2}.$

\begin{figure}[t!]
\centering
\includegraphics[width=0.28\textwidth]{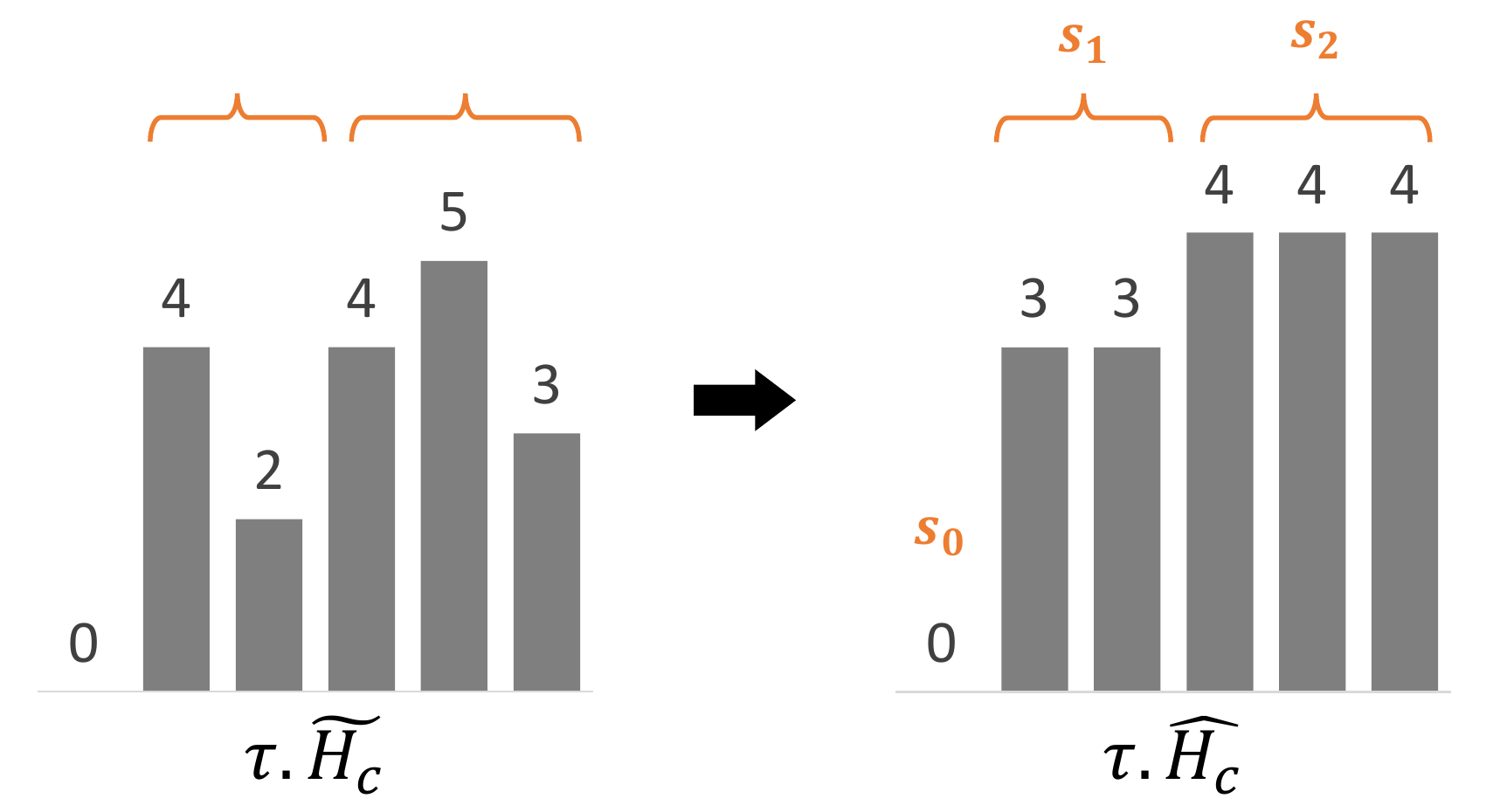}
\vspace{-9pt}
\caption{$L_2$ Isotonic Regression converts a noisy histogram that is no longer nondecreasing into a nondecreasing histogram by partitioning and averaging within each partition.}
\label{fig:l2}
\end{figure}

\subsubsection{Variance for the cumulative histogram method} \label{subsubsec:HcVar}
The cumulative histogram method also uses isotonic regression, but uses it to create $\tau.\widehat{H}_c$ while we need estimates of the variance of $\tau.\widehat{H}_g[i]$. Since the $H_g$ representation is a nonlinear transformation of the $H_c$ representation, we must use a different way to estimate variance.

 We know that before the isotonic regression in the $H_c$ method, independent noise with scale $1/\epsilon_1$ was added to each cell of $\tau.H_c$ so we (over)estimate the variance of $\tau.\widehat{H}_c[j]$ as $2/\epsilon_1^2$. The estimated number of groups of size $j$ is $\tau.\widehat{H}_{c}[j]-\tau.\widehat{H}_{c}[j-1]$ so it would have variance $4/\epsilon_1^2$. Dividing by the number of estimated groups of size $j$, this gives $4\epsilon_1^{-2}/(\tau.\widehat{H}_{c}[j]$ $-\tau.\widehat{H}_{c}[j-1])$ as the estimated variability in each group whose size is estimated to be $j$. Thus we set our estimate of the variance of the $i^\text{th}$ largest group to be the following in Line \ref{line:estvar} of Algorithm \ref{alg:top-down}:
$$\textstyle{\tau.V_g[i]=4 / (\epsilon_1^2\times\text{number of estimated groups of size $\tau.\widehat{H}_g[i]$})}$$

\subsection{Optimal Matching}
\label{subsec:match}
Now, every group belongs to a node in each level of the hierarchy (e.g., a household in Fairfax county is also a household in Virginia). Since group size distributions are estimated independently at all levels of the hierarchy (lines \ref{line:tdlevel} -- \ref{line:tdswitch} in Algorithm \ref{alg:top-down}), this causes several problems. First, each group has several size estimates (a household has a certain estimated size from the Fairfax County estimates and another estimated size from the Virginia estimate). Second, from these separate estimates, we don't know which group in one level of the hierarchy corresponds to which group at a different level of the hierarchy (e.g., is it possible that the $102^\text{nd}$ largest household in Virginia, with estimated size 12 is the same as the $3^\text{rd}$ largest household in Fairfax County, with estimated size 13?).

Thus, to make the group size \changes{distribution} estimates consistent, we need to estimate a matching between groups in $\tau$ and groups in the children of $\tau$, as in Figure \ref{fig:illustration}. 
For privacy reasons, such a matching must only be done using the differentially private data generated in Line \ref{line:tdswitch} of Algorithm \ref{alg:top-down}. In this section we explain how to perform this matching and in the next section we explain how to reconcile the different size estimates.

Formally, a matching is a function $m$ that inputs a node $\tau$ and an index $i$ and returns a node $c$ and index $j$ with the semantics that the $i^\text{th}$ smallest group in $\tau$ is believed to be the same as the $j^\text{th}$ smallest group in child $c$ of $\tau$. We must estimate this matching using only differentially private data (e.g., $\tau.\widehat{H}$ for each node $\tau$). We first convert each $\tau.\widehat{H}$ into the \changes{unattributed} representation $\tau.\widehat{H}_g$, where $\tau.\widehat{H}_g[i]$ is the size of the $i^\text{th}$ smallest group in $\tau$.

For each node $\tau$, we set up a bipartite weighted graph as shown in Figure \ref{fig:illustration}. There are $\tau.G$ nodes on the top half. We label them as $(\tau, 1), (\tau,2),\dots, (\tau, \tau.G)$. Each node on the bottom half has the form $(c, j)$, where $c$ is a child of $\tau$ and $j$ is an index into $c.\widehat{H}_g$. Between every node $(\tau, i)$ and $(c, j)$ there is an edge with weight $|\tau.\widehat{H}_g[i]-c.\widehat{H}_g[j]|$ which measures the difference in estimated size between the $i^\text{th}$ smallest group in $\tau$ and the $j^\text{th}$ smallest group in $c$. 

Our desired matching is then the least cost weighted matching on this bipartite graph. Sophisticated matching algorithms (e.g., based on network flows) can find an optimal matching, but they have time complexity at least $O(\tau.G)^3$ \cite{lovasz}. In our case, $\tau.G$ can be in the millions (e.g., there are over 100 million households in the U.S.).

There exists a well-known 2-approximation algorithm for matching, which adds edges to the matching in order of increasing weight \cite{lovasz}. However on our graph it would run in $O(\tau.G^2\log\tau.G)$ time as there are $\tau.G^2$ edges and they need to be sorted. Instead, we take advantage of the special weight structure of our edges to produce an optimal matching algorithm with time complexity $O(\tau.G\log\tau.G)$.

\begin{figure}[t!]
\centering
\includegraphics[width=0.38\textwidth]{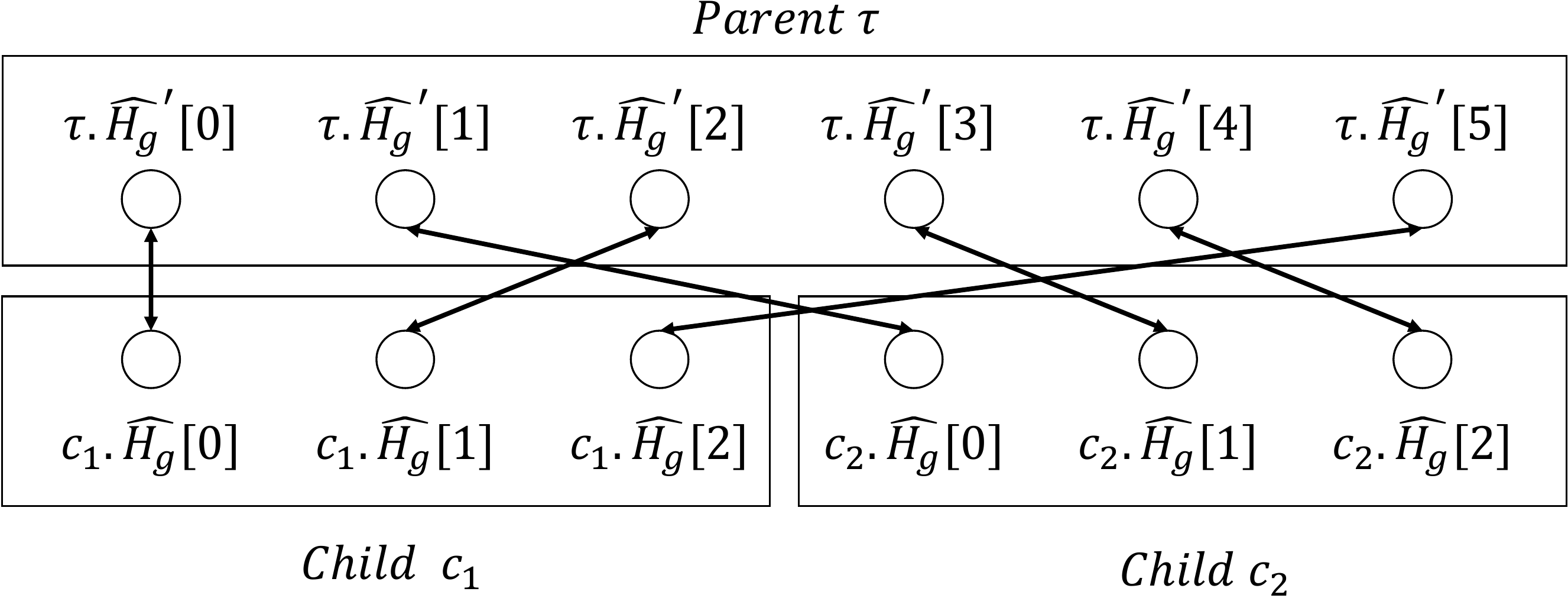}
\vspace{-9pt}
\caption{Consistency Matching Illustration}
\label{fig:illustration}
\end{figure}


\subsubsection{Optimal Matching Algorithm}
\begin{algorithm}[t!]
\caption{Matching}
\label{alg:opt}
\DontPrintSemicolon
\SetNoFillComment
\KwIn{$\tau.\widehat{H}_g$ and $c.\widehat{H}_g$ for $c\in\child(\tau)$}
\tcp{Unmatched nodes from top of bipartite graph}
Top $\gets \set{i~:~i=1,\dots,\tau.G}$\;
\tcp{Unmatched nodes from the bottom}
Bot $\gets\set{(c,i)~:~c\in\child(\tau); i=1,\dots, c.G}$\;
\While{Top$\neq \emptyset$}{
    \tcp{Smallest unmatched group size in Top}
    $s_t \gets \min\set{\tau.\widehat{H}_g[i]~:~i\in \text{Top}}$\;
    \tcp{Smallest unmatched group size in Bot}
    $s_b \gets \min\set{c.\widehat{H}_c[i]~:~(c,i)\in \text{Bot}}$\;
    \tcp{All groups from Top with min size}
    $G_t \gets \set{i\in\text{Top}~:~\tau.\widehat{H}_g[i]=s_t}$\;
    \tcp{All groups from Bot with min size}
    $G_b \gets \set{(c,i)\in\text{Bot}~:~c.\widehat{H}_g[i]=s_b}$\;
    \uIf{$|G_t|\geq |G_b|$}{
        \tcp{all nodes in $G_b$ can be matched now}
            For each $i\in G_t$, assign it arbitrarily to a unique $(c,j)\in G_b$\label{line:optarb1}\;
                 Remove $i$ from Top and $(c,j)$ from Bot\;
        }
       \Else{
           \tcp{Some nodes in $G_b$ can be matched now}
            num$[c]\gets$ \# records in $G_b$ belonging to child $c$.\;
           \For{each child $c$ of $\tau$} {
                    assign $|G_t|\frac{\text{num}[c]}{\sum_c \text{num}[c]}$ of the nodes in $G_t$ to arbitrary nodes in $G_b$ from child $c$\;\label{line:optarb2a}
                    Remove matched nodes from Top and Bot\;
                 }
            }
     }
    \Return the matching
\end{algorithm}

The algorithm is shown in Algorithm \ref{alg:opt}. To achieve the desired time complexity, we sort the groups in $\tau$ in increasing order and do the same to the set of all of groups in all of the children (hence the $O(\tau.G\log\tau.G)$ cost). We then proceed by matching the smallest unmatched group in $\tau$ to the smallest unmatched group among any of its children.

Normally there are many groups from $\tau.\widehat{H}_g$ with the same size and many child nodes with the same size that they can match to. For example, there can be 300 groups of size $1$ in $\tau.\widehat{H}_g$ and the child nodes together can have 200 groups of size $1$. Thus we can match 200 of the groups from the parent with the 200  groups at the children. The specifics of which group of size 1 in the parent matches which group of size 1 in the children is completely unimportant (since the groups of size 1 at the parent are completely indistinguishable from each other) and hence can be done arbitrarily as in Line \ref{line:optarb1} in Algorithm \ref{alg:opt}. After this matching, the remaining 300-200= 100 groups of size $1$ in the parent will then be matched with groups of size 2 in the child nodes, etc.

 Sometimes, not all of the groups in the children can be matched at once. For example, the parent may have 300 groups of size $1$, but there are 3 children with child $c_1$ having 200 groups of size 1, $c_2$ having 100 groups of size $1$, and $c_3$ having $100$ groups of size 1. In this case, the assignments are done proportionally:\footnote{Sometimes the proportions tell us that the $r$ groups at the parent should be matched with $r_1$ groups at child $c_1$, $r_2$ at child $c_2$, etc. (with $r=r_1+r_2+\dots$) but the $r_i$ are real numbers instead of integers. In this case we find the unique k such that rounding up the $r_i$ with the k largest fractional parts and rounding the rest down gives integers that sum up to $r$.}
  50\% of the parent's groups of size 1 are matched to the corresponding groups in $C_1$, 25\% are matched to the corresponding groups at $c_2$, and the remaining 25\% are matched to the groups in $c_3$ (e.g., line \ref{line:optarb2a} in Algorithm \ref{alg:opt}). In this situation, all of the size 1 groups in the parent will have been matched, but there will be some size 1 groups in the children that are not yet matched. The next iteration of the algorithm will try to match them to size 2 groups in the parent.
  
  \begin{lemma}
  If the weight between edges of the form $(\tau, i)$ and $(c,j)$ equals $|\tau.\widehat{H}_g[i]-c.\widehat{H}_g[j]|$, then Algorithm \ref{alg:opt} finds the optimal least-cost perfect matching.
  \end{lemma}
  \begin{proof}
  We say that two matchings $m_1$ and $m_2$ have a \emph{trivial} difference at edge $(a,b)$ if all of the following hold:
  \begin{enumerate}[noitemsep,leftmargin=5mm]
  \item $m_1$ matches $a$ to $b$.
  \item there are nodes $d,c$ such that $m_2$ matches $a$ to $c$ and $d$ to $b$.
  \item if $m_2$ is modified to match $a$ to $b$ and $d$ to $c$, then the cost of $m_2$ doesn't change.
  \end{enumerate}
  Two matchings $m_1$ and $m_2$ have a \emph{non-trivial} difference at edge $(a,b)$ if $(a,b)$ is part of matching $m_1$ but not $m_2$, and it is not a trivial difference.
  
  Let $m$ be the matching returned by Algorithm \ref{alg:opt}. Assume, by way of contradiction, that $m$ is not optimal. In this case, there will always be optimal matchings with no trivial differences with $m$ (i.e. all differences will be nontrivial). This is because we can make trivial differences at an edge $(a,b)$ disappear by performing the modification discussed above which doesn't cause the cost to change.
  
  When examining the order in which pairs of matched nodes are added to $m$ by Algorithm \ref{alg:opt}, for any optimal matching $m^\prime$ with no trivial differences, there is a first time at which $m$ and $m^\prime$ have a non-trivial difference, e.g., they agree on the first $n$ matchings but disagree on the $n+1\text{st}$. Let $m^*$ be an optimal matching that has no trivial differences, and has the largest possible time for its first non-trivial difference. Let the edge on which $m$ and $m^*$ first differ be $((\tau, i), (c, j))$. 
 This means that $m^*$ matches $(\tau,i)$ to some $(c^*, j^*)$ and matches some $(\tau, i^*)$ to $(c,j)$. 
 
 By construction of the algorithm, the following are true:
 \begin{align*}
 \tau.\widehat{H}_g[i] &\leq \tau.\widehat{H}_g[i^*]\quad\text{and}\quad
 c.\widehat{H}_g[j] \leq c^*.\widehat{H}_g[j^*],
\end{align*}
otherwise either node $(c^*,j^*)$ would have been matched by the algorithm before $(c,j)$, triggering a non-trivial difference earlier, or $(\tau, i^*)$ would have been matched by the algorithm before $(\tau, i)$, causing a non-trivial difference earlier; in either case, this would contradict the fact that $((\tau, i), (c, j))$ is the first non-trivial difference. 

 Now there are only four possible cases: 
 \begin{enumerate}[noitemsep]
 \item $\tau.\widehat{H}_g[i] \leq \tau.\widehat{H}_g[i^*] \leq c.\widehat{H}_g[j] \leq c^*.\widehat{H}_g[j^*]$ 
\item $\tau.\widehat{H}_g[i] \leq c.\widehat{H}_g[j] \leq \tau.\widehat{H}_g[i^*]  \leq c^*.\widehat{H}_g[j^*]$ 
\item $c.\widehat{H}_g[j] \leq c^*.\widehat{H}_g[j^*] \leq \tau.\widehat{H}_g[i] \leq \tau.\widehat{H}_g[i^*]$ 
\item $c.\widehat{H}_g[j]\leq \tau.\widehat{H}_g[i] \leq c^*.\widehat{H}_g[j^*]  \leq \tau.\widehat{H}_g[i^*]$ 
\end{enumerate}
 
 Cases 1 and 3 are symmetric (they involve interchanging the top and bottom of the bipartite graph) and Cases 2 and 4 are also symmetric in the same way. Thus the same proof technique for Case 1 will apply to Case 3, and the same technique for Case 2 will apply to Case 4.
 
 \noindent\textbf{Case 1: }$\tau.\widehat{H}_g[i] \leq \tau.\widehat{H}_g[i^*] \leq c.\widehat{H}_g[j] \leq c^*.\widehat{H}_g[j^*]$. Based on this ordering, we first list a few identities:
 \begin{align}
 |\tau.\widehat{H}_g[i] - c.\widehat{H}_g[j]| &= \substack{|\tau.\widehat{H}_g[i] - \tau.\widehat{H}_g[i^*]| \\+ |\tau.\widehat{H}_g[i]-c.\widehat{H}_g[j]|}\label{eqn:optproof1}\\
 |\tau.\widehat{H}_g[i^*]  - c^*.\widehat{H}_g[j^*]| &= \substack{|\tau.\widehat{H}_g[i^*] - c.\widehat{H}_g[j]| \\+ |c.\widehat{H}_g[j]- c^*.\widehat{H}_g[j^*]|} \label{eqn:optproof2}\\
 |\tau.\widehat{H}_g[i] - c^*.\widehat{H}_g[j^*]|&=\substack{|\tau.\widehat{H}_g[i] - \tau.\widehat{H}_g[i^*]|\\+|\tau.\widehat{H}_g[i]-c.\widehat{H}_g[j]|\\+|c.\widehat{H}_g[j]- c^*.\widehat{H}_g[j^*]|}\label{eqn:optproof3}\\
 |\tau.\widehat{H}_g[i^*]-c.\widehat{H}_g[j]| &= |\tau.\widehat{H}_g[i^*]-c.\widehat{H}_g[j]|\label{eqn:optproof4}
 \end{align}
 
 Now we see that the sum of Equations \ref{eqn:optproof1} and \ref{eqn:optproof2} equals the sum of Equations \ref{eqn:optproof3} and \ref{eqn:optproof4}.  Thus we can change the optimal matching $m^*$ by matching $(\tau, i)$ to $(c,j)$ (instead of the original connection $(\tau, i)$ to $(c^*, j^*)$) and match $(\tau, i^*)$ to $(c^*,j)$ and the cost of the matching will stay the same (contradicting the choice of $m^*$ -- that it is not supposed to have trivial differences from $m$; in fact, such reassignment of edges is how trivial differences are removed).

 \noindent\textbf{Case 2:} $\tau.\widehat{H}_g[i] \leq c.\widehat{H}_g[j] \leq \tau.\widehat{H}_g[i^*]  \leq c^*.\widehat{H}_g[j^*]$
 
 From this ordering it is clear that: 
 \begin{align*}
 \lefteqn{|\tau.\widehat{H}_g[i] - c.\widehat{H}_g[j]| + |\tau.\widehat{H}_g[i^*]  - c^*.\widehat{H}_g[j^*]|}\\
 &\leq |\tau.\widehat{H}_g[i] - c^*.\widehat{H}_g[j^*]|\\
 &\leq |\tau.\widehat{H}_g[i] - c^*.\widehat{H}_g[j^*]| + |\tau.\widehat{H}_g[i^*]-c.\widehat{H}_g[j]|
  \end{align*}
  Thus we can change the optimal matching $m^*$ by matching $(\tau, i)$ to $(c,j)$ (instead of the original connection $(\tau, i)$ to $(c^*, j^*)$) and match $(\tau, i^*)$ to $(c^*,j)$ and the cost of the matching will either decrease (contradicting optimality of $m^*$) or stay the same (contradicting the choice of $m^*$ -- that it is not supposed to have trivial differences from $m$).
  
\noindent\textbf{Case 3:} $c.\widehat{H}_g[j] \leq c^*.\widehat{H}_g[j^*] \leq \tau.\widehat{H}_g[i] \leq \tau.\widehat{H}_g[i^*]$. Symmetric to case 1, so omitted.

\noindent\textbf{Case 4:} $c.\widehat{H}_g[j]\leq \tau.\widehat{H}_g[i] \leq c^*.\widehat{H}_g[j^*]  \leq \tau.\widehat{H}_g[i^*]$. Symmetric to case 2, so omitted.
 \end{proof}

\subsection{Merging Estimates}\label{subsec:merge}
%
Given a node $\tau$, the matching algorithm assigns one group in $\tau$ to one group in some child of $\tau$ (i.e. it says that, for every $i$, the $i^\text{th}$ smallest group in $\tau$ matches the $j^\text{th}$ smallest group in child $c$ of $\tau$). This means that for every group, we have two estimates of its size: $\tau.\widehat{H}_g[i]$ and $c.\widehat{H}_g[j]$ as well as corresponding estimates of its variance $\tau.V_g[i]$ and $c.V_g[j]$.
There are two possible ways of reconciling these estimates.

\textbf{Naive strategy.}
The simplest way is to simply average  $\tau.\widehat{H}_g[i]$ and $c.\widehat{H}_g[j]$. This approach would be valid if the variance estimates were not accurate -- recall that it is not possible to estimate the variances exactly, so they needed to be approximated. However, as we show in our experiments, a weighted averaging based on the estimated variance outperforms this strategy.

\textbf{Variance-based weighted strategy.} If we have two noisy estimates of the same quantity (e.g., $\tau.\widehat{H}_g[i]$ and $c.\widehat{H}_g[j]$), along with their variances  $\tau.\widehat{H}_g[i]$ and $c.\widehat{H}_g[j]$, it is a well-known statistical fact (e.g., see \cite{Hay2010Boosting}) that the optimal linear way of combining the estimates is to estimate the size as a weighted average, where the weights are inversely proportional to the variance:
{
\begin{align}
\textstyle{
\left(\frac{\tau.\widehat{H}_g[i]}{\tau.V_g[i]} + \frac{c.\widehat{H}_g[j]}{c.V_g[j]}\right)  \Big/\left(\frac{1}{\tau.V_g[i]} + \frac{1}{c.V_g[j]}\right)\label{eqn:varest1}
}
\end{align}
}

and the variance of this estimator is
{
\begin{align}
\textstyle{
\left(\frac{1}{\tau.V_g[i]} + \frac{1}{c.V_g[j]}\right)^{-1}\label{eqn:varest2}
}
\end{align}
}
Thus we would update the size of the $j^\text{th}$ largest group at child $c$ using Equation \ref{eqn:varest1} and its new variance as Equation \ref{eqn:varest2}. This estimator is preferable to the naive strategy when the variance estimates are accurate.

The size estimates are then rounded. After the size estimates at the children are updated, the top-down algorithm continues matching the groups in each child $c$ with the groups at the children of $c$. Once the groups at the leaves are updated, the resulting sizes at the leaves are treated as the final estimates.

%
%
%
%
%
 
\subsection{Privacy}
\begin{theorem}
Algorithm \ref{alg:top-down} satisfies $\epsilon$-differential privacy.
\end{theorem}
\begin{proof}
Privacy is easy to analyze because the algorithm separates the differentially private data access from the post-processing. Specifically, the only part of the algorithm that touches the sensitive data occurs in Lines \ref{line:tdlevel} through \ref{line:tdswitch}. It uses sequential composition across levels of the hierarchy. Thus each of the $L+1$ levels is assigned $\epsilon/(L+1)$ of the privacy budget. Within each level there is parallel composition because adding or removing one person from a group only affects the node that contains that group (and none of the sibling nodes). The \changes{\coco} histograms produced at each node use either the method of Section \ref{subsec:unatt} or Section \ref{subsec:cumu}) with privacy budget $\epsilon/(L+1)$ and they scale the noise correctly to the global sensitivity, as discussed in those sections.

The rest of the top-down algorithm is completely based on the differentially private results of of Lines \ref{line:tdlevel} through \ref{line:tdswitch}. The conversions between $\tau.\widehat{H}_g, \tau.\widehat{H}_c,$ and $\tau.\widehat{H}$ are trivial manipulations of the histogram format (they do not touch the original data). The variance estimation is based on $\tau.\widehat{H}_g, \tau.\widehat{H}_c,$ and $\tau.\widehat{H}$. The matching algorithm only uses $\tau.\widehat{H}_g$ for each node, and the method for merging estimates only uses $\tau.\widehat{H}_g$ (for each node) and the associated variance estimates (which were computed from $\tau.\widehat{H}_g, \tau.\widehat{H}_c,$ and $\tau.\widehat{H}$). Since post-processing differentially private outputs still satisfies differential privacy \cite{DMNS06Calibrating}, the overall algorithm satisfies differential privacy.
\end{proof}

\section{Experiments}\label{sec:experiments}
In this section, we present our experiments, which were conducted on a machine with dual 8-core 2.1 GHz Intel(R) Xeon(R) CPUs and 64 GB RAM. 

\subsection{Datasets}
We used 4  large-scale datasets for evaluation.

\textbf{\underline{Partially synthetic housing}}. Individuals live in households and group quarters. The number of individuals in each facility is important but this information was truncated past households of size 7 in the 2010 Decennial Census, Summary File 1 \cite{sf1}. Thus we created a partially synthetic dataset that mirrors the published statistics, but adds a heavy tail as would be expected from group quarters (e.g., dormitories, barracks, correctional facilities). This was done for each state by estimating the ratio $\#$ households of size 7/$\#$ households of size 6, and then randomly sampling (with a binomial distribution) the number of groups of size $k\geq 8$ so that the same ratio holds (in expectation) between number of groups of neighboring sizes. Then 50 outliers groups are chosen with size uniformly distributed between 1 and 10,000. The hierarchy in this levels are National and State (50 states plus Puerto Rico and the District of Columbia). The third level is County, which we obtained by randomly assigning groups at the state to their counties (the assignment was proportional to county size).

\textbf{\underline{NYC taxi}}: We use 143,540,889 Manhattan taxi trips from the 2013 New York City taxi dataset \cite{nyctaxi}. An anonymized taxi medallion (e.g., a taxi) is considered a group and the size of the group is the number of pickups it had in a region. The region hierarchy is the following. Level 0: Manhattan; level 1: upper/lower Manhattan; level 2: 28 neighborhoods from NTA boundary \cite{nycopen}.

\textbf{\underline{Race distribution}} (white and Hawaiian): For each block, based on 2010 Census data (in Summary File 1 \cite{sf1}), we count the number of whites and number of native Hawaiians that live in the block. Hence block is treated as a group. The hierarchy is National, State, and County. We performed evaluations on all 6 major race categories recorded by the Census, but omitted the rest due to space restrictions.  

The statistics for our datasets are the following.\\
\begin{tabular}{ | c | c | c | c |  }
\hline
   Data      &   \# groups  & \# people/trip  &   \# unique size  \\
\hline 
Synthetic & 240,908,081 & 605,304,918 & 2352 \\
\hline
White & 11,155,486 & 226,378,365 & 1916 \\
\hline
Hawaiian & 11,155,486 & 540,383 & 224\\
\hline
Taxi & 360,872 & 130,962,398 & 3128 \\
\hline
\end{tabular}
For some \changes{\coco} estimation methods, such as the cumulative sum $H_c$ method, one needs to specify a public maximum group size $K$. We set $K=100,000$ as a conservative estimate (for example, in our partially synthetic housing dataset, the true max size was around $10,000$, an order of magnitude smaller). 
For our 2-level hierarchy experiments on Census related data, we used National/State. For 3 level, we used West Coast/State/County. All numbers plotted are averaged over 10 runs.

%

\begin{figure*}[ht]
	\centering
	\includegraphics[width=14.5cm]{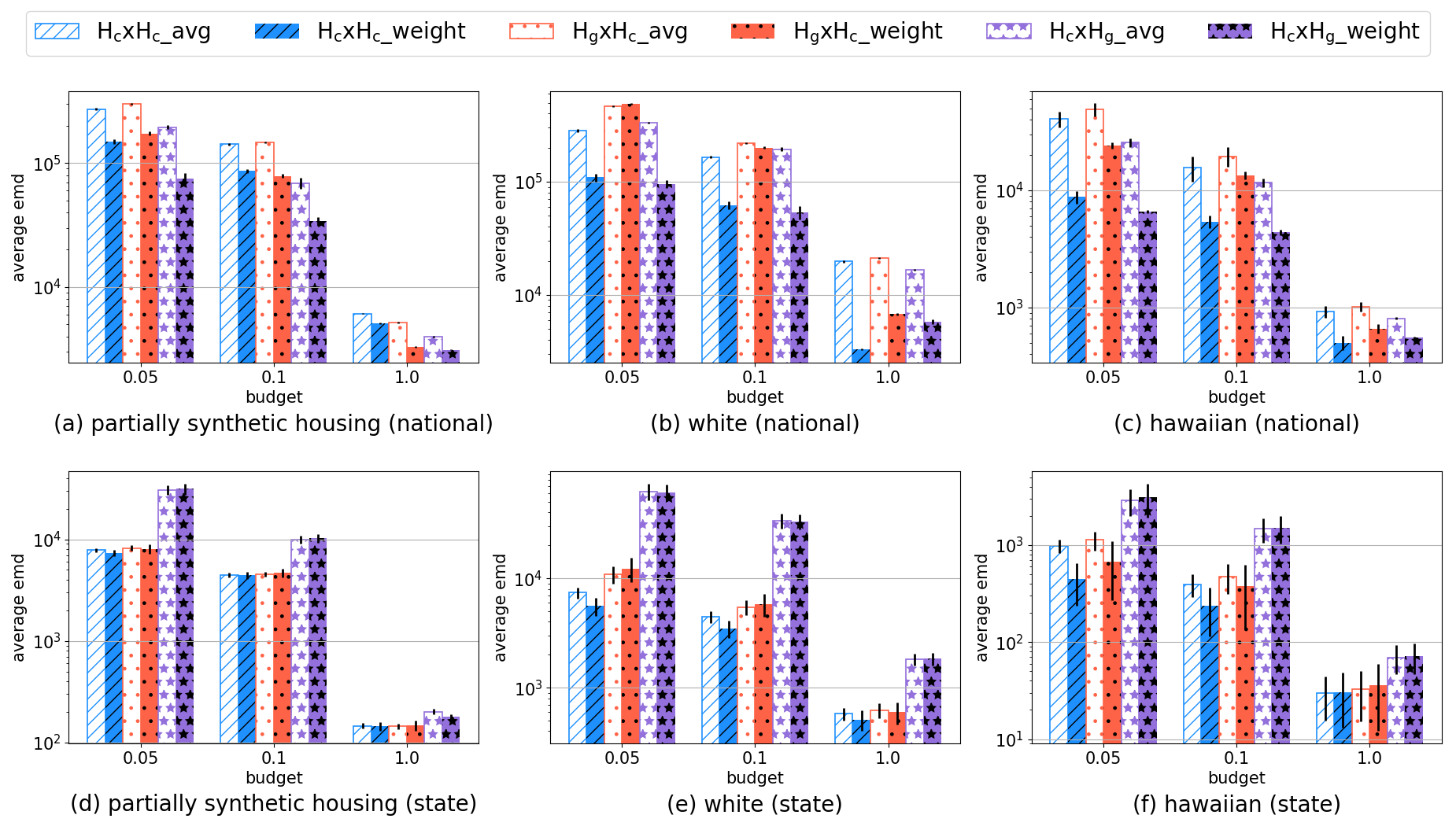}
 \vspace{-7pt}
	\caption{Merging estimates using weighted average vs. normal average. x-axis: privacy budget per level.}
  \label{fig:metrics}
\end{figure*}

\subsection{Evaluation}
\textbf{Evaluation metric}. For each level of the hierarchy, we evaluate Earthmover's distance (emd) as discussed in \Cref{subsec:error}, per node in the level, in order to see error at each level. We do not aggregate error across all levels of the hierarchy as there is no principled way of, for example, weighting the importance of error at the state level compared to the national level. Each error measure is averaged across 10 runs. The standard deviation of the average is then the empirical standard deviation (which gives std for one random run) divided by $\sqrt{10}$ (which then gives the std for the mean of 10 runs). We plot 1 std error bars on each figure.

\textbf{Algorithm selection}. We evaluate a variety of choices for generating hierarchical \changes{\coco histograms}. For example, we can use the $H_g$ method to generate estimates at national level but $H_c$ at state level and $H_c$ again at county level, which we would denote as $H_g\times H_c\times H_c$ and compare it against $H_g\times H_g\times H_g$. \changes{In general, using the $H_c$ method at each level gives best results except on the partially synthetic data, thus we recommend it as the default choice. However, as is true with all data-dependent, no algorithm is expected to dominate on all datasets \cite{Kotsogiannis2017:Pythia}. For fine-grained algorithm selection, one could use Pythia \cite{Kotsogiannis2017:Pythia} or the approach of Chaudhuri et al. \cite{Chaudhuri2011ERMJMLR}. As these are well-established techniques, they are outside the scope of the paper.}

\textbf{Interpreting error}. How can we interpret the error numbers of the algorithms? I.e., what is a good error? We suggest comparison against the following ``omniscient'' algorithm. Given an $\epsilon$, the omniscient algorithm will know which group sizes exist in which node in the hierarchy and reduce it to the problem of estimating a simple histogram of known group size by location. Thus it will split its privacy budget per level of the hierarchy and will add Laplace$(1/\epsilon)$ noise (with standard deviation of $\sqrt{2}/\epsilon$) only to those groups that exist. Meanwhile, the $\epsilon$-differentially private algorithms must effectively estimate which group sizes exist and also estimate their sizes, along with which nodes in the hierarchy the groups belong to. The error for the omniscient algorithm is then expected to be $[\#\text{distinct group sizes} \times \sqrt{2}/\epsilon\times \#\text{levels}]$. For example, in Figure \ref{fig:metrics}, at privacy budget $0.1$ per level and the partially synthetic housing data at the national level, there are $2,352$ groups and so the omniscient algorithm will have expected error around $3.3\times 10^4$, which is in line with the results for the best differentially private method in that figure.

\subsubsection{Ruling out the naive method}\label{subsec:expnaive}
The naive strategy adds noise to $\tau.H$ directly (Section \ref{subsec:naive}).   Its average error with $\epsilon=1$ is \changes{in the billions, as} shown in the following table \changes{confirming the analysis in \Cref{subsec:naive}.} 

\smallskip

\hspace{-14pt}
\begin{tabular}{ | c | c | c | c |  }
\hline
 Synthetic    &  White  & Hawaiian  & Taxi  \\
\hline
4,462,728,374 &  4,809,679,734 & 4,027,891,692 & 208,977,518\\
\hline
\end{tabular}
\changes{This error is several orders of magnitude larger than the other methods (e.g. Figure \ref{fig:metrics}), so is not considered further.}


\begin{figure*}[htb!]
	\centering
    \includegraphics[width=14.5cm]{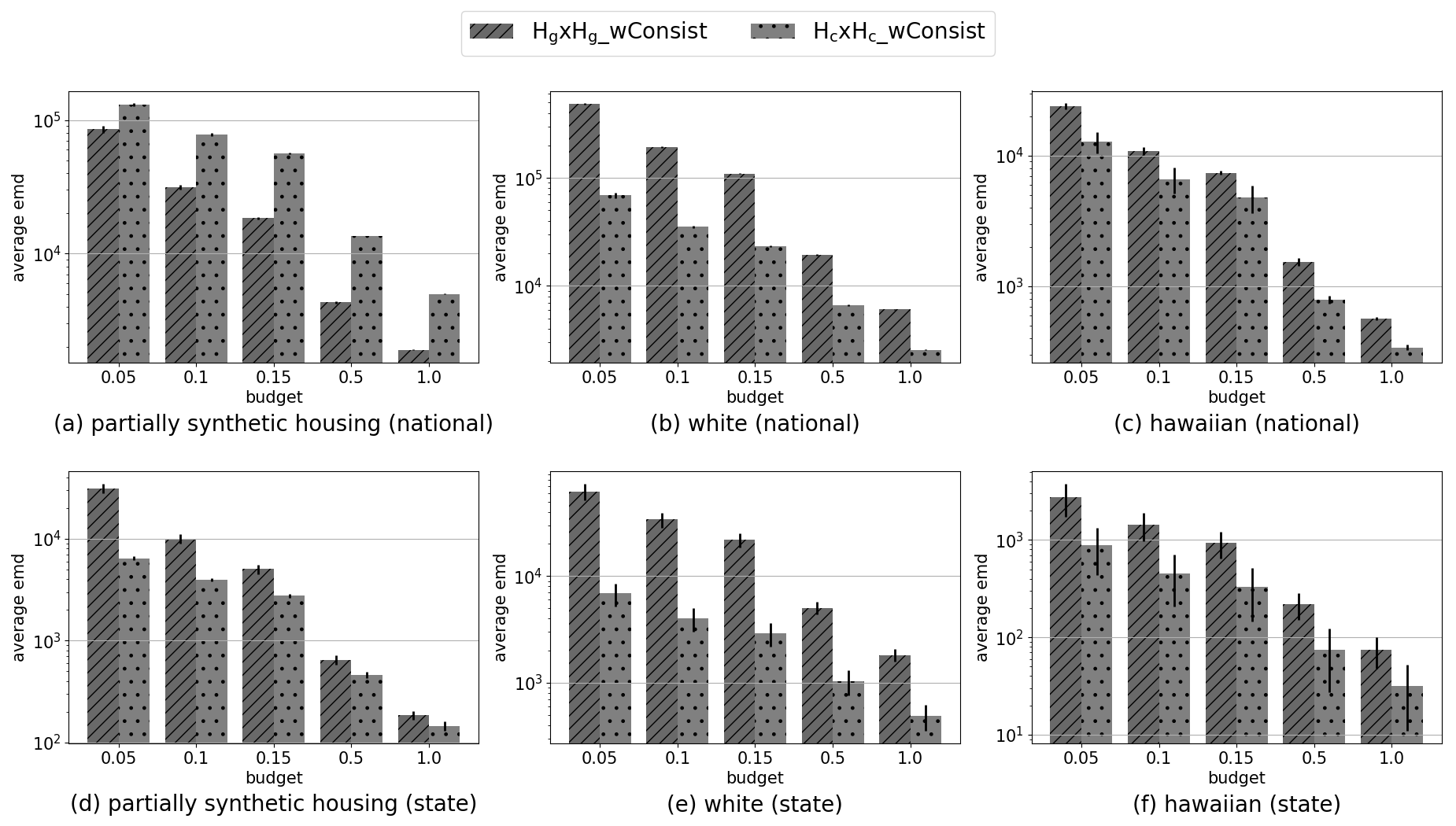}
\vspace{-5pt}
	\caption{2-level consistency at each level. x-axis: privacy budget per level.}
  \label{fig:2lv-all}
\end{figure*}

\subsubsection{Comparison to bottom-up aggregation}
\changes{
The Bottom-Up (BU) baseline allocates all privacy budget to the leaves and then sets the \coco histogram of a parent to be the sum of the histograms at the leaves. As mentioned earlier, it is expected to have very low error at the leaves but higher error everywhere else (an observation that is true in general for hierarchical problems \cite{Qardaji2013Hierarchical,Hay2010Boosting}). In the following table we examine the error at each level of BU compared to the consistency algorithm that uses $H_c$ at each level. Both approaches use a total privacy budget of $\epsilon=1.0$.
}

\vspace{-0.2em}
\begin{tabular}{c|r|r|r|r|}\cline{2-5}
\multicolumn{1}{c}{} & \multicolumn{1}{|c}{Part. Synth.} & \multicolumn{1}{|c}{White}& \multicolumn{1}{|c}{Hawaiian} & \multicolumn{1}{|c|}{Taxi}\\\cline{2-5}
\multicolumn{0}{c}{}&\multicolumn{4}{c}{Level 0}\\\cline{2-5}
BU &  $78,459.0$ & $448,909.0$ & $13,968.0$ & $20,731.0$\\
$H_c$ & $32,480.0$ & $17,000.0$ & $1,381.0$ & $10,547.0$\\\cline{2-5}
\multicolumn{0}{c}{}&\multicolumn{4}{c}{Level 1}\\\cline{2-5}
BU &  $1,512.2$ & $8,722.0$ & $270.1$ & $10,405.5$\\
$H_c$ & $ 1,000.3$ & $1,511.8 $ & $ 117.7$ & $ 5,431.5$\\\cline{2-5}
\multicolumn{0}{c}{}&\multicolumn{4}{c}{Level 2}\\\cline{2-5}
BU &  $24.9 $ & $152.3$ & $4.3$ & $772.8$\\
$H_c$ & $ 80.1$ & $363.8$ & $21.6$ & $1,601.8$\\\cline{2-5}
\end{tabular}

\changes{The results are as expected, with significant improvements at the higher level nodes in exchange for a small increase in error at the leaves.}

\subsubsection{Weighted Average Estimation Comparison} \label{subsec:wavg}

Next we evaluate the merging strategy for groups that are matched across different levels of the hierarchy (as described in Section \ref{subsec:merge}). There are two choices. If a group from one level of the hierarchy (e.g., a household at the state level) is matched to a group at a different level (e.g., at the national level), we can either average the two group size estimates, or use a weighted average based on their estimated variances from Section \ref{subsec:initvar}. If the variance estimates are accurate, we expect such a merging to result in improved histogram error when compared to normal averaging (see Figure \ref{fig:metrics}).

We estimate two levels of the hierarchy for each dataset (partially synthetic, Hawaiian race distribution, White race distribution are shown due to space limitations) and consider various methods for the initial group size estimates in each of the two levels (Section \ref{sec:nonhier}); for example $H_c\times H_g$ means that the $H_c$ method was used for the top level and $H_g$ for the second level. We see that the weighted average method consistently produces large reductions in error (compared to normal averaging) at the top level for every privacy budget and combination of methods ($H_g/H_c$), and consistently produces modest improvements in error at the second level. Hence we conclude that our variance estimates, which are used in the merging procedure, are useful approximations. The only method missing from these graphs is the combination $H_g\times H_g$. On these datasets, the error of normal average was so large that it would visually skew the results. Due to the superiority of weighted averaging, all following experiments use this method of merging estimates.

\subsubsection{2-Level Hierarchy Results} \label{subsec:2lv}


We next show our results for the case when  2 levels of the hierarchy need to be estimated, in  \Cref{fig:2lv-all}. Due to space limitations, we present the consistency result for $H_g\times H_g$ and $H_c\times H_c$ using weighted averaging to merge estimates. One thing to notice is that the best performing method is comparable to the omniscient baseline \changes{and generally the $H_c$ method is the one that performs best.} Typically, we expect the $H_c$ method to dominate for ``dense'' data. For example, white population data contains many groups from size 0 to size 3000. On the other hand, the partially synthetic housing data is more sparse at the national level, with many small groups (e.g., sizes 1-12) followed by large gaps between group sizes. In such cases, methods based on $H_g$ work better.

\subsubsection{3-Level Hierarchy Results}  \label{subsec:3lv}

%
%
%

\begin{figure*}[ht]
	\centering
	\includegraphics[width=18cm]{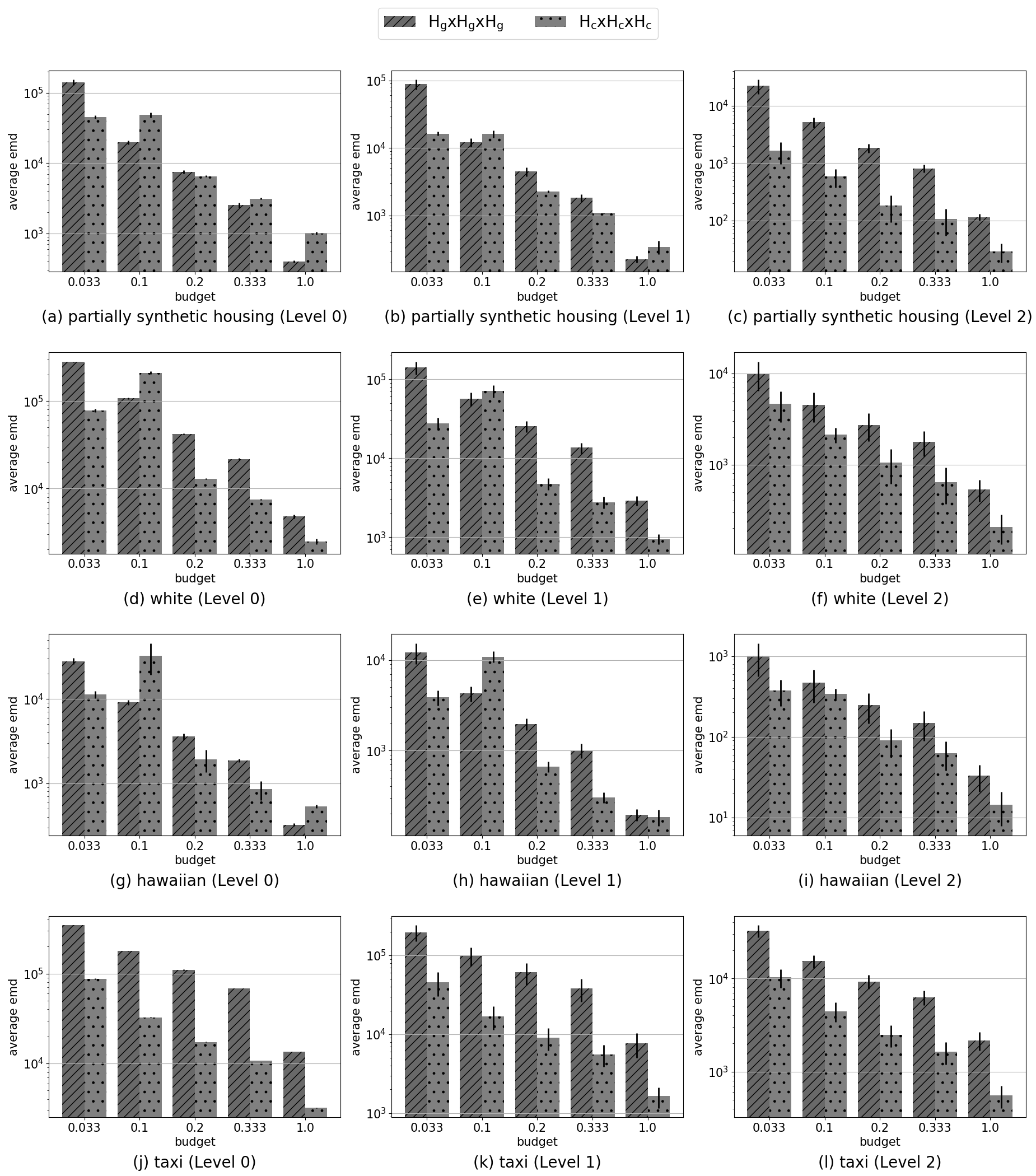}
	\caption{3-level consistency at each level. x-axis: privacy budget per level.}
  \label{fig:3lv-all}
\end{figure*}

We next show the results when three levels of the hierarchy need to be estimated. Because there are over 3,000 counties (hence 3,000 isotonic regressions), for computational reasons we limit the hierarchy to the west coast (for the partially synthetic housing data as well as the race distribution data.). Of these two alternatives $H_g \times H_g \times H_g$ and $H_c \times H_c \times H_c$ on synthetic, race and taxi at each level in \Cref{fig:3lv-all}. The taxi data uses its full geography. In general, we see that no method dominates the other, but group size estimation methods based on $H_c$ generally perform better than $H_g$ and so are  a good default choice.  

\section{Conclusions and Future Work}\label{sec:future}
In this paper we introduced the differentially private \emph{hierarchical} \changes{\coco} histogram problem and presented a solution based on isotonic regression and optimal weighted matchings. This problem is motivated by a variety of tables that are published in truncated form in Summary File 1 of the 2010 U.S. Decennial Census. The actual tables include additional demographic characteristics that are attached to the household sizes at each level of geography. Such additional information greatly expands the dimensionality of the problem and is an area of future work.


\smallskip
\section{Acknowledgments}
We thank John Abowd for discussions on disclosure avoidance requirements for publishing Census tables. \changes{We thank the anonymous reviewers for their help comments}. The work was partially supported by NSF grants 1054389, 1544455, and 1702760. The views in this paper are those of the authors, and do not necessarily represent the Census Bureau.
\balance

\clearpage
\bibliographystyle{abbrv}
\bibliography{ref}  


\end{document}